\documentclass[preprint,11pt]{imsart}

\RequirePackage[numbers]{natbib}
\RequirePackage[colorlinks,citecolor=blue,urlcolor=blue]{hyperref}
\usepackage[utf8]{inputenc}
\usepackage[margin=3cm]{geometry}
\usepackage{amsmath,amssymb,amsthm,graphicx}
\usepackage{algorithm}
\usepackage{algpseudocode}
\usepackage{dsfont} 
\usepackage{booktabs}
\usepackage{caption}
\usepackage{subcaption}
\usepackage{tikz}
\usetikzlibrary{shapes,decorations,arrows,calc,arrows.meta,fit,positioning}
\usepackage{multirow} 
\usepackage{color}

\newtheorem{theorem}{Theorem}[section]

\newtheorem{lemma}[theorem]{Lemma}

\newcommand{\R}{\mathbb{R}}
\newcommand{\E}{\mathbb{E}}
\DeclareMathOperator*{\argmin}{arg\,min}

\begin{document}
\begin{frontmatter}

\title{Dual Likelihood for Causal Inference under Structure Uncertainty}
\runtitle{Dual Likelihood for Causal Inference under Structure Uncertainty}

\author{\fnms{David} \snm{Strieder}\corref{}\ead[label=e1]{david.strieder@tum.de}}
\and
\author{\fnms{Mathias} \snm{Drton}\ead[label=e2]{mathias.drton@tum.de}}
\address{Munich Center for
Machine Learning and Department of Mathematics, \\ TUM School of Computation, Information and Technology, Technical University of Munich \\
	\printead{e1,e2}}

\runauthor{D. Strieder and M. Drton}

\begin{abstract}%
Knowledge of the underlying causal relations is essential for inferring the effect of interventions in complex systems. In a widely studied approach, structural causal models postulate noisy functional relations among interacting variables, where the underlying causal structure is then naturally represented by a directed graph whose edges indicate direct causal dependencies. In the typical application, this underlying causal structure must be learned from data, and thus, the remaining structure uncertainty needs to be incorporated into causal inference in order to draw reliable conclusions. In recent work, test inversions provide an ansatz to account for this data-driven model choice and, therefore, combine structure learning with causal inference. In this article, we propose the use of dual likelihood to greatly simplify the treatment of the involved testing problem. Indeed, dual likelihood leads to a closed-form solution for constructing confidence regions for total causal effects that rigorously capture both sources of uncertainty: causal structure and numerical size of nonzero effects. The proposed confidence regions can be computed with a bottom-up procedure starting from sink nodes. To render the causal structure identifiable, we develop our ideas in the context of linear causal relations with equal error variances.
\end{abstract}

\begin{keyword}[class=MSC]
\kwd[Primary ]{62D20, 62H22}
\end{keyword}

\begin{keyword}
\kwd{linear structural causal models}
\kwd{causal effects}
\kwd{dual likelihood}
\kwd{uncertainty quantification}
\kwd{graphical models}
\end{keyword}

\end{frontmatter}

\section{Introduction}
Reasoning about causal relations and inferring the effect of interventions in complex systems is a central task of scientific research. The field of causal discovery addresses this challenge with much fundamental research being done over the last decades \citep{pearl:book, spirtes:book}. In particular, in many applied settings, access to interventional data is limited, and thus, performing classical controlled experiments to investigate causal relations is not feasible. To tackle this problem, various identifiability results clarify under which conditions interventional distributions can theoretically be identified from observational data alone. Furthermore, many structure learning algorithms have been proposed that estimate causal structures using only observed data. On the other hand, when the underlying causal structure is known, causal inference results enable researchers to estimate causal quantities in complex systems and assess remaining uncertainty. However, when the underlying causal structure is learned from data, the remaining structure uncertainty needs to be incorporated into causal inference in order to draw reliable conclusions.

Recently, \citet{strieder:2021} introduced a first ansatz for constructing confidence regions for causal effects that account for both types of uncertainty, uncertainty about the causal structure and uncertainty about the numerical size of the effect. They consider bivariate linear structural causal models (SCMs) and develop a solution based on test inversion.  Their strategy was improved and generalized to higher-dimensional linear SCMs in \citet{strieder:2023}.   In this work, we take up their idea of using a test inversion ansatz to construct confidence regions for total causal effects (see Section \ref{section:background}).  However, in contrast to the earlier work, we suggest here to employ dual likelihood for the arising testing problem.  An important consequence of the use of dual likelihood is that we are able to provide a closed-form solution for our proposed confidence regions (Section \ref{section:result}) and can, thus, avoid the numerical optimization algorithms deployed by \citet{strieder:2023}. Furthermore,  while the earlier methods rely on a top-down procedure starting from the source node, by employing dual likelihood, we obtain a bottom-up procedure starting from the sink node. Thus, our proposed dual likelihood method provides an interesting alternative perspective on strategies to save on computation. Finally, in Section \ref{section:simulation}, we present computational details and analyze the performance in a simulation study. Further simulation results can be found in Appendix \ref{appendix:simulation}. In the concluding Section \ref{section:discussion}, we discuss extensions and possible future work.

\section{Background}\label{section:background}  
A common tool for studying causal relationships among interacting variables are structural causal models \citep{peters:book, handbook}, where each variable is represented as a function of other variables (its causes) and a random error term. The causal perspective views these relationships as assignments rather than mathematical equations, and thus, changes in the causes result in changes in the effects but not vice versa, reflecting the inherent asymmetry in cause-effect relationships. In this section, we review a restricted class of causal models, namely linear structural causal models with Gaussian errors and equal error variances \citep{peters:2014}, as well as the definition of the total causal effect, which is the causal target of interest in this study. Further, we review the test inversion ansatz introduced by \citet{strieder:2021} to construct confidence regions for causal effects that account for structure uncertainty.

\subsection{Causal Effects in Linear Structural Causal Models}
We assume access to observational data in the form of $n$ independent copies of a random vector $X = (X_1, . . . , X_d)$ with zero mean. To ensure unique identifiability, we follow a line of research introduced by \citet{peters:2014} that focuses on linear relations and normal distributed errors with equal variances, represented by the equation system
\begin{equation} \label{eq:LSEM}
    X_j=\sum_{i\neq j}\beta_{j,i}X_i + \varepsilon_j, \quad j=1, \dots ,d,
\end{equation}
where $B:=[\beta_{j,i}]_{j,i=1}^d$ represents the direct causal effects between variables, and $\varepsilon_j$ are independent, normally distributed error terms with common variance $\sigma^2$. Such a linear structural causal model is naturally represented by a directed graph, where a missing edge $i\rightarrow j$ indicates no direct effects, that is, $b_{j,i}=0$. Further, we assume the underlying graph to be acyclic, which entails the unique solution $X=(I_d-B)^{-1} \varepsilon$ of the equation system \eqref{eq:LSEM}, where $I_d$ denotes the $d\times d$ identity matrix. The corresponding covariance matrix is given by $\Sigma=\sigma^2(I_d-B)^{-1}(I_d-B)^{-T}$.

In the remainder of the article, we use the following notation and graphical concepts. We write $i<_{G}j$ if node $i$ precedes node $j$ in a causal ordering of the corresponding directed acyclic graph (DAG). If the DAG contains an edge from node $i$ to node $j$, then node $i$ is called a parent of node $j$, and we denote the set of all parents of node $j$ with $p(j)$. Further, if the DAG contains a directed path from node $i$ to node $j$, then node $j$ is called a descendant of node $i$, and we denote the set of all descendants of node $i$ with $d(i)$. Finally, we write $\Sigma_{j,i|p(i)}$ for the conditional covariance matrix, that is,
    \begin{equation*}
        \Sigma_{j,i|p(i)}:=\Sigma_{j,i}-\Sigma_{j,p(i)}(\Sigma_{p(i),p(i)})^{-1}\Sigma_{p(i),i}.
    \end{equation*}

Our aim is to estimate the total causal effect $\mathcal{C}(i\rightarrow j)$ in linear SCMs and rigorously quantify the remaining uncertainty. Formally, the total causal effect is defined as the unit change in the expectation of $X_j$ with respect to an intervention in $X_i$
\begin{equation*} 
   \mathcal{C}(i \rightarrow j):=\frac{\text{d}}{\text{d} x_i} \E[X_j|\text{ do}(X_i=x_i)]=\Sigma_{j,i|p(i)}(\Sigma_{i,i|p(i)})^{-1}.
\end{equation*}
While this parameter of interest is given as a simple function of the covariance matrix, difficulties arise in practice when the underlying causal structure is unknown. In such situations, the conditioning set is unknown and has to be inferred from data, which introduces structure uncertainty. We emphasize that due to the identifiability of Gaussian SCMs with equal error variances \citep[e.g.,][]{chen:2019,ghoshal:2018}, estimating the total causal effect is nevertheless a well-defined problem. 

A naive two-step approach that treats the two tasks of causal discovery and causal inference separately by learning the causal structure first and then calculating confidence intervals in the inferred model does not account for uncertainty in the structure. Furthermore, classical bootstrapping methods also fail to correctly account for model uncertainty due to singularities at the intersection of models. To address this issue, \citet{strieder:2021} introduced a framework for constructing confidence intervals that rigorously account for structure uncertainty, which was then generalized in \cite{strieder:2023}. However, thoroughly accounting for structure uncertainty proves a challenging task, given the enormous number of possible structures, even in moderate dimensions. 

\subsection{Confidence Intervals via Test Inversion}

The main idea is to leverage the classical duality between statistical hypothesis tests and confidence regions \citep{casella:book}.  To this end, we  consider all possible causal structures in the following way. If we perform valid hypothesis tests for all attainable causal effects, then all values that cannot be rejected form a confidence region for the total causal effect. This shifts the task to constructing suitable hypothesis tests for all attainable causal effects. 

Under the assumption of an underlying linear structural causal model with equal error variances, every possible distribution $N(0,\Sigma)$ is given by a set of covariance matrices $\mathcal{M}:=\bigcup_{G \in \mathcal{G}(d)} \mathcal{M}(G)$, where $\mathcal{G}(d)$ is the set of all possible complete DAGs on $d$ nodes and
\begin{equation*}
    \mathcal{M}(G)=\Big\{\Sigma \in \text{PD}(d): \exists  \sigma^2 >0 \text{ with } 
	    \sigma^2=\Sigma_{k,k|p(k)} \quad  \forall\  k=1, \dots , d \Big\}.
\end{equation*}
Here, we write $\text{PD}(d)$ for the cone of positive definite $d \times d$ matrices. The hypothesis of a fixed causal effect $\mathcal{C}(i \rightarrow j)$ of size $\psi$ further restricts the set of possible distributions, where the specific structure of the constraint depends on the DAG $G$. Thus, with
\begin{equation}\label{eq:hypothesis}
	     \mathcal{M}_{\psi}(G) := \Big\{ \Sigma \in \text{PD}(d) : \exists \sigma^2>0 \text{ with } 
	    \psi \sigma^2=\Sigma_{j,i|p(i)} \text{ and }
	    \sigma^2=\Sigma_{k,k|p(k)} \quad  \forall\  k=1, \dots , d
	    \Big\},
	\end{equation}
and $\mathcal{M}_\psi:=\bigcup_{G \in \mathcal{G}(d)} \mathcal{M}_\psi(G)$, the task is to invert the statistical testing problem 
\begin{align}\label{eq:testproblem}
    H_0^{(\psi)} :
	    \Sigma \in \mathcal{M}_\psi \quad \text{against} \quad H_1 : \Sigma \in \mathcal{M} \backslash \mathcal{M}_\psi
\end{align}
for all $\psi \in \R$.  

\section{Testing for Total Causal Effects with the Dual Likelihood}\label{section:result}

Earlier methods \citep{strieder:2023} propose to use constrained likelihood ratio tests for the testing problem $\eqref{eq:testproblem}$. However, the direct maximization of the Gaussian likelihood with constraints on the total causal effect has no closed-form solution, which leads to algorithms that rely on numerical optimization routines and grid searches. Instead, our proposal is to employ dual likelihood theory for Gaussian models \citep[e.g.,][]{brown:1986,kauermann:1996} to solve the testing problem $\eqref{eq:testproblem}$ and, thus,  obtain a simple closed-form solution.

To obtain our main result, we first introduce the dual likelihood for Gaussian graphical models. Maximizing the Gaussian likelihood corresponds to the minimization problem
\begin{equation*}
    \argmin_{\Sigma \in \mathcal{M}} \mathrm{KL}(P_{\widehat{\Sigma}},P_{\Sigma})=\argmin_{\Sigma \in \mathcal{M}} \big(\mathrm{tr}(\Sigma^{-1}\widehat{\Sigma})+\log \det(\Sigma)\big),
\end{equation*}
where KL is the Kullback–Leibler divergence with sample covariance $\widehat{\Sigma}$. As the Kullback–Leibler divergence is not symmetrical in its arguments, the minimization problem 
\begin{align*}
    \argmin_{\Sigma \in \mathcal{M}} \mathrm{KL}(P_{\Sigma},P_{\widehat{\Sigma}})&=\argmin_{\Sigma \in \mathcal{M}}\big(\mathrm{tr}(\widehat{\Sigma}^{-1}\Sigma)-\log \det(\Sigma) \big) \\
    &=\argmin_{\Sigma \in \mathcal{M}}\big(\mathrm{tr}(\Sigma\widehat{\Sigma}^{-1})+\log \det(\Sigma^{-1}) \big) 
\end{align*}
yields a different estimator, the dual maximum likelihood estimator. Along similar lines, the dual (log-)likelihood of $\Sigma$ for the observed parameter $\widehat{\Sigma}^{-1}$ is defined as
\begin{equation*}
    \tfrac{2}{n}\ell^{dual}_n(\Sigma|\widehat{\Sigma}^{-1}):=-\log \det (2\pi \Sigma^{-1}) - \mathrm{tr}(\Sigma\widehat{\Sigma}^{-1}).
\end{equation*}
We note that this dual likelihood is reciprocal to the Gaussian likelihood in the sense that for observed sample covariance $\widehat{\Sigma}^{-1}$, we have 
\begin{equation}\label{eq:equal}
    \ell^{dual}_n(\Sigma|\widehat{\Sigma}^{-1})=\ell_n(\Sigma^{-1}|\widehat{\Sigma}^{-1}),
\end{equation}
where $\tfrac{2}{n}\ell_n(\Sigma|\widehat{\Sigma}):=-\log \det (2\pi \Sigma) - \mathrm{tr}(\Sigma^{-1}\Hat{\Sigma}).$

\subsection{Dual Likelihood Ratio Test}
In order to construct confidence intervals for the total causal effects, our idea is to invert dual likelihood ratio tests of the testing problem \eqref{eq:testproblem} with the dual likelihood ratio test statistic 
\begin{align} \label{eq::dualstat}
    \text{dual-}\check{\lambda}^{(\psi)}_n:&=2\Big(\sup_{\Sigma \in  \mathcal{M}} \ell^{dual}_n(\Sigma|\widehat{\Sigma}^{-1})- \sup_{\Sigma \in  \mathcal{M}_{\psi}} \ell^{dual}_n(\Sigma|\widehat{\Sigma}^{-1})\Big)  \\
    &=2\Big(\sup_{\Sigma \in \mathcal{M}} \ell_n(\Sigma^{-1}|\widehat{\Sigma}^{-1})- \sup_{\Sigma \in  \mathcal{M}_{\psi}} \ell_n(\Sigma^{-1}|\widehat{\Sigma}^{-1})\Big). \nonumber
\end{align}
First, we observe that by employing equation \eqref{eq:equal}, this dual likelihood ratio test statistic equals the classical likelihood ratio test statistic for a modified problem. More precisely, we redefine the model space $\mathcal{M}^{-1}:=\bigcup_{G \in \mathcal{G}(d)} \mathcal{M}^{-1}(G)$ with
\begin{equation*}
    \mathcal{M}^{-1}(G):=\Big\{\Omega \in \text{PD}(d): \exists \tilde{\sigma}^2 >0 \text{ with } 
	     \tilde{\sigma}^2 =\Omega_{k,k|d(k)} \quad  \forall\  k=1, \dots , d \Big\},
\end{equation*}
where $\Omega=\Sigma^{-1}$. Similarly, we redefine the hypothesis space $\mathcal{M}^{-1}_\psi=\bigcup_{G \in \mathcal{G}(d)} \mathcal{M}^{-1}_\psi(G)$ with
\begin{align*}
    \mathcal{M}^{-1}_\psi(G):=\Big\{\Omega \in \text{PD}(d): \exists \tilde{\sigma}^2>0 \text{ with } 
	    \tilde{\sigma}^2 &=\Omega_{k,k|d(k)} \quad  \forall\  k=1, \dots , d \\ \text{ and }
	   \psi &=\tau_j\big(\Omega_{i,d(i)}(\Omega_{d(i),d(i)})^{-1}\big) \Big\},
\end{align*}
where $\tau_j$ projects the $|d(i)|$-dimensional vector onto the component corresponding to $j$ if $j\in d(i)$ and zero otherwise. Then,  we obtain the following Lemma.
\begin{lemma}\label{lemma::set} Let $\Sigma \in \text{PD}(d)$. Then the following equivalences hold.
\begin{enumerate}
    \item $\Sigma \in \mathcal{M}$ if and only if $\Sigma^{-1} \in \mathcal{M}^{-1}$.
    \item $\Sigma \in \mathcal{M}_{\psi}$ if and only if $\Sigma^{-1} \in \mathcal{M}^{-1}_{\psi}$.
\end{enumerate}
\end{lemma}

\begin{proof}
This result follows from straightforward calculations using the following observation. If $\Sigma = \sigma^2(I_d-B)^{-1}(I_d-B)^{-T}$ for some DAG $G$ with edge weights $B$ and equal error variance $\sigma^2 >0$, then $\Sigma^{-1} = \sigma^{-2}(I_d-T)^{-1}(I_d-T)^{-T}$, where $T$ represents the negative total causal effect between variables. Thus, $\Sigma^{-1}$ corresponds to the covariance matrix of a DAG $G'$, where the parent sets $p(i)$ in $G'$ are the descendants $d(i)$ in $G$ for all $i= 1, \dots , d$, the edge weights are given by $T$ and equal error variance $\sigma^{-2}$. Constraints on the total causal effects in $\Sigma$ and $G$ thus correspond to constraints on direct effects in $\Sigma^{-1}$ and $G'$.
\end{proof}

Using Lemma \ref{lemma::set}, we immediately obtain that the dual likelihood ratio test statistic \eqref{eq::dualstat} corresponds to 
\begin{align*}
    \text{dual-}\check{\lambda}^{(\psi)}_n&=2\Big(\sup_{\Omega \in \mathcal{M}^{-1}} \ell_n(\Omega|\widehat{\Sigma}^{-1})- \sup_{\Omega \in  \mathcal{M}^{-1}_{\psi}} \ell_n(\Omega|\widehat{\Sigma}^{-1})\Big),
\end{align*}
which is the classical likelihood ratio test statistic for the modified testing problem 
\begin{equation}\label{eq:testproblem2}
 \tilde{H}_0^{(\psi)} :
	    \Omega \in \mathcal{M}^{-1}_\psi \quad \text{against} \quad \tilde{H}_1 : \Omega \in \mathcal{M}^{-1} \backslash \mathcal{M}^{-1}_\psi,
\end{equation}
with observed sample covariance $\widehat{\Sigma}^{-1}$. Thus, we can employ similar strategies to \citet{strieder:2023} in order to obtain confidence regions for total causal effects.

\subsection{Confidence Intervals with Dual Likelihood}

In this subsection, we state our main result, a confidence region for total causal effects $\mathcal{C}(i\rightarrow j)$ that captures structure uncertainty as well as uncertainty about the numerical size of the effect. Our first step to tackle the testing problem \eqref{eq:testproblem2} is to relax the alternative $\mathcal{M}^{-1}$ and employ the theory of intersection union tests, see e.g.~\cite{casella:book}, to obtain a simple upper bound on the distribution of the dual likelihood ratio test statistic via
\begin{equation*}
    \text{dual-}\check{\lambda}^{(\psi)}_n\leq \text{dual-}\lambda^{(\psi)}_n(G):= 2\Big(\sup_{\Omega \in \text{PD}(d)} \ell_n(\Omega|\widehat{\Sigma}^{-1})- \sup_{\Omega\in  \mathcal{M}^{-1}_{\psi}(G)} \ell_n(\Omega|\widehat{\Sigma}^{-1})\Big).
\end{equation*}
Further, we define $\text{dual-}\check{\lambda}_n^{(\psi)}(i <_G j):=\min_{ G \in \mathcal{G}(d)\,:\, i <_G j} \text{dual-}\check{\lambda}_n^{(\psi)}(G)$, where
\begin{equation*}
     \text{dual-}\check{\lambda}^{(\psi)}_n(G):=2\Big(\sup_{\Omega \in \mathcal{M}^{-1}} \ell_n(\Omega|\widehat{\Sigma}^{-1})- \sup_{\Omega \in  \mathcal{M}^{-1}_{\psi}(G)} \ell_n(\Omega|\widehat{\Sigma}^{-1})\Big).
\end{equation*}

Then we obtain our following main result via test inversions of asymptotically conservative hypothesis test based on the asymptotic distribution of the upper bound $\text{dual-}\lambda^{(\psi)}_n(G)$ under every single hypothesis $\mathcal{M}^{-1}_\psi(G)$, see Lemma \ref{theorem:asymptotics}, as well as the closed-form solution of the dual likelihood estimation under constraints on the total causal effect, see Section \ref{section:estimation}. 

\begin{theorem}\label{theorem:duallrtinterval}
    Let $\alpha \in (0,1)$. Then an asymptotic $(1-\alpha)$-confidence set for the total causal effect $\mathcal{C}(i\rightarrow j)$ is given by
    \begin{align*}
            C := \{\psi \in \R  :  \text{\normalfont{dual-}}\check{\lambda}_n^{(\psi)}(i <_G j)\leq \chi_{d,1-\alpha}^2 \} \cup \{0 : \text{\normalfont{dual-}}\check{\lambda}_n^{(0)}(j <_G i)\leq \chi_{d-1,1-\alpha}^2\}.
    \end{align*}
    This confidence set can be constructed explicitly as follows. Define 
    \begin{align*}
        K&:=\min_{G \in \mathcal{G}(d)}\sum_{k=1}^{d}(\widehat{\Sigma}^{-1})_{k,k|d(k)}, &
        Z&:=\min_{G \in \mathcal{G}(d): j <_G i}\sum_{k=1}^{d}(\widehat{\Sigma}^{-1})_{k,k|d(k)},
    \end{align*}
    and for all $G \in \mathcal{G}(d)$ with $i <_G j$, 
    \begin{equation*}
        D_G:= (\widehat{\Sigma}^{-1})_{i,j|d(i)\setminus \{j\}}^2-(\widehat{\Sigma}^{-1})_{j,j|d(i)\setminus \{j\}}\Big(\sum_{k\neq i}^{d}(\widehat{\Sigma}^{-1})_{k,k|d(k)} + (\widehat{\Sigma}^{-1})_{i,i|d(i)\setminus \{j\}} -K\exp\big(\tfrac{\chi^2_{d,1-\alpha}}{dn}\big)\Big).
    \end{equation*}
    Moreover, for $D_G \geq 0$, define
    \begin{equation*}
        L_G:=\frac{-(\widehat{\Sigma}^{-1})_{i,j|d(i)\setminus \{j\}}  - \sqrt{D_G}}{(\widehat{\Sigma}^{-1})_{j,j|d(i)\setminus \{j\}}},  \qquad U_G:=\frac{-(\widehat{\Sigma}^{-1})_{i,j|d(i)\setminus \{j\}}  + \sqrt{D_G}}{(\widehat{\Sigma}^{-1})_{j,j|d(i)\setminus \{j\}}}.
    \end{equation*}
    Then the closed-form solution for the confidence set $C$ is given by
    \begin{equation*}
         C= \bigcup_{G\in \mathcal{G}(d): D_G \geq 0 }\big[L_G,U_G\big] \, \bigcup \, \big\{ 0 : Z \leq K\exp\big(\tfrac{\chi^2_{d-1,1-\alpha}}{dn}\big) \big\}.
    \end{equation*}
\end{theorem}
\begin{proof}
     Let $\psi \in \R$ and $\Sigma \in \mathcal{M}_{\psi}$. Since $\mathcal{M}_{\psi}=\bigcup_{G \in \mathcal{G}(d)} \mathcal{M}_{\psi}(G)$ there exists a complete graph $G$, such that  $\Sigma \in \mathcal{M}_{\psi}(G)$. If $i <_G j$, it follows with the introduced upper bound and Lemma \ref{theorem:asymptotics}
    \begin{equation*}
        P_{\Sigma}\Big( \text{dual-}\check{\lambda}^{(\psi)}_n > \chi_{d,1-\alpha}^2 \Big)   \leq P_{\Sigma}\Big( \text{dual-}\lambda^{(\psi)}_n (G) > \chi_{d,1-\alpha}^2 \Big) \rightarrow \alpha.
    \end{equation*} 
    Furthermore, plugging in the dual likelihood estimates \eqref{eq:unconstrdual} and \eqref{eq:constrdual} derived in Section \ref{section:estimation}, we view
    \begin{equation*}
        \text{dual-}\check{\lambda}^{(\psi)}_n(G) - \chi_{d,1-\alpha}^2=2\Big(\sup_{\Omega \in \mathcal{M}^{-1}} \ell_n(\Omega|\widehat{\Sigma}^{-1})- \sup_{\Omega \in  \mathcal{M}^{-1}_{\psi}(G)} \ell_n(\Omega|\widehat{\Sigma}^{-1})\Big)-\chi_{d,1-\alpha}^2
    \end{equation*}
     as a strictly convex quadratic polynomial in $\psi$. With $K:=\min_{G \in \mathcal{G}(d)}\sum_{k=1}^{d}(\widehat{\Sigma}^{-1})_{k,k|d(k)}$ and 
    \begin{equation*}
        D_G:= (\widehat{\Sigma}^{-1})_{i,j|d(i)\setminus \{j\}}^2-(\widehat{\Sigma}^{-1})_{j,j|d(i)\setminus \{j\}}\Big(\sum_{k\neq i}^{d}(\widehat{\Sigma}^{-1})_{k,k|d(k)} + (\widehat{\Sigma}^{-1})_{i,i|d(i)\setminus \{j\}} -K\exp\big(\tfrac{\chi^2_{d,1-\alpha}}{dn}\big)\Big),
    \end{equation*}
    this quadratic polynomial has real roots if $D_G \geq 0$. Thus, the inequality $\text{dual-}\check{\lambda}^{(\psi)}_n(G) \leq \chi^2_{d,1-\alpha}$ holds if and only if
    \begin{equation*} 
        \psi \in \Bigg[\frac{-(\widehat{\Sigma}^{-1})_{i,j|d(i)\setminus \{j\}}  - \sqrt{D_G}}{(\widehat{\Sigma}^{-1})_{j,j|d(i)\setminus \{j\}}}, \frac{-(\widehat{\Sigma}^{-1})_{i,j|d(i)\setminus \{j\}}  + \sqrt{D_G}}{(\widehat{\Sigma}^{-1})_{j,j|d(i)\setminus \{j\}}}  \Bigg].
    \end{equation*}
    We obtain an analogous result for $j <_G i$, and the test inversion approach yields the claim.
\end{proof}

Note that for non-zero effects, we do not need to test among graphs with $j <_G i$.  Thus, our confidence regions may consist of a non-zero interval and an isolated zero, reflecting the remaining uncertainty about the direction of the causal relation.  In contrast, for zero-sized effects, we need to consider all causal orderings since multiple paths might cancel. 

Our main result employs the following asymptotic distribution of the upper bound $\text{dual-}\lambda^{(\psi)}_n(G)$ under every single hypothesis $\mathcal{M}^{-1}_\psi(G)$.

\begin{lemma}\label{theorem:asymptotics}
	Let $G \in \mathcal{G}(d)$ be a DAG and let $\psi \in \R$. Then, the relaxed dual likelihood ratio test 
        statistic $\text{\normalfont{dual-}}\lambda^{\psi}_n(G)$ satisfies one of the following properties.
        \begin{itemize}
	    \item[a)] If $i <_G j$, then under the hypothesis $H_0^{(\psi)}(G)$
                \begin{equation*}
	           \text{\normalfont{dual-}}\lambda^{(\psi)}_n(G) \overset{\mathcal{D}}{\rightarrow}        \chi^2_d.
	        \end{equation*}
            \item[b)] If $j <_G i$, then under the hypothesis $H_0^{(0)}(G)$
                \begin{equation*}
	          \text{\normalfont{dual-}}\lambda^{(0)}_n(G) \overset{\mathcal{D}}{\rightarrow}        \chi^2_{d-1}.
	        \end{equation*}
	\end{itemize}
\end{lemma}

\begin{proof}
    We assume $i <_G j $. Without loss of generality, let $(1,2, \dots, d)$ be the causal ordering of $G$. Then $ \mathcal{M}^{-1}_\psi(G) = \{ \Omega \in \text{PD}(d) : f_\psi(\Omega|G)=0\}$, with
    \begin{equation*}
	    f_\psi(\Omega|G):=\begin{pmatrix}
	    \tau_j(\Omega_{i,\{i+1, \dots , d\}}(\Omega_{\{i+1, \dots , d\},\{i+1, \dots , d\}})^{-1}) - \psi\\ 
	    \Omega_{1,1|\{2, \dots , d\}}-\Omega_{d,d} \\
	    \vdots \\
	    \Omega_{d-1,d-1|d}-\Omega_{d,d}
	    \end{pmatrix}.
    \end{equation*}

    The Jacobian of $f_\psi$ has full rank $d$, which follows due to its (nonzero) triangular structure in derivatives for $\Omega_{k,k}$ for all $k=1, \dots , d-1$, while the first row has nonzero derivative for $ \Omega_{i,j}$ but zero derivatives for $ \Omega_{k,k}$ for all $k=1, \dots , i$. Thus, $\mathcal{M}^{-1}_\psi(G)$ defines a $\frac{1}{2}(d^2-d)$-dimensional submanifold of $\R^{\frac{1}{2}(d^2+d)}$. The case $j <_G i $ follows similarly, and thus, the claim follows. For details, we refer to \citet{drton:2009}.
\end{proof}

Note that the constraint of a zero-sized effect does not restrict the submanifold $\mathcal{M}^{-1}_0(G)$ for $j <_G i $. All edge weights in the corresponding linear SCMs can vary freely, 
 and thus, the two different degrees of freedom for the chi-square limits arise.

\subsection{Dual Likelihood Estimation}\label{section:estimation}

Maximizing the classical Gaussian likelihood under constraints on total causal effects leads to complex optimization problems with polynomial constraints that need to be solved with numerical optimization routines. In contrast to that, we can explicitly maximize the dual likelihood under constraints on total causal effects. Thus, as an important consequence of the use of dual likelihood, we are able to provide a closed-form solution for our proposed confidence regions in Theorem \ref{theorem:duallrtinterval}. 

In order to derive this closed-form solution, we need to calculate the dual likelihood ratio test statistic \eqref{eq::dualstat}. This involves maximizing the dual likelihood in two cases, the unrestricted case $\Sigma \in \mathcal{M}$, and under the constraint of a fixed-size total causal effect $\Sigma \in \mathcal{M}_{\psi}$. In the following, we solve the dual likelihood equations for a fixed graph $G$. In order to consider full uncertainty over all possible causal structures, we subsequently need to cleverly optimize over the space of DAGs, see Section \ref{section:shortcuts}. In the case $\Sigma \in  \mathcal{M}(G)$, we obtain
    \begin{align*}
        \tfrac{2}{n}\sup_{\Sigma \in  \mathcal{M}(G)} \ell^{dual}_n(\Sigma|\widehat{\Sigma}^{-1})&=\sup_{\Omega \in \mathcal{M}^{-1}(G)}-\log \det (2\pi \Omega) - \mathrm{tr}(\Omega^{-1}\Hat{\Sigma}^{-1}) \\
        &=\sup_{T \in \R^{-G}, \sigma^2>0} -d\log(2\pi \sigma^{-2}) - \sigma^2\mathrm{tr}\big((I-T)^{T}(I-T)\Hat{\Sigma}^{-1}\big),
    \end{align*}
    where $\R^{-G}:=\{T \in \R^{d \times d} : t_{k,l}=0 \text{ if } l \notin d(k) \}$. This immediately follows if we recapitulate 
    \begin{equation*}
        \Sigma = \sigma^2(I_d-B)^{-1}(I_d-B)^{-T} = \sigma^2 (I-T)^T (I-T),
    \end{equation*}
    where $T \in \R^{-G}$ represents the negative total causal effect between variables. Maximizing over the equal error variance $\sigma^2$ then leads to a linear least squares problem with the optimal solution
    \begin{equation}\label{eq:unconstrdual}
        \min_{T \in \R^G}\mathrm{tr}\big((I-T)^{T}(I-T)\Hat{\Sigma}^{-1}\big)= \sum_{k=1}^{d}(\widehat{\Sigma}^{-1})_{k,k|d(k)}.
    \end{equation}
In the second case of an additional constraint of a fixed size total causal effect $\mathcal{C}(i\rightarrow j)=\psi$, we can similarly first optimize over the equal error variance $\sigma^2$ to derive a linear least squares problem. That is, for $\Sigma \in \mathcal{M}_{\psi}(G)$,  maximizing the dual likelihood then corresponds to the linear least squares problem
    \begin{equation*}
        \min_{T \in \R^{-G,\psi}}\mathrm{tr}\big((I-T)^{T}(I-T)\Hat{\Sigma}^{-1}\big).
    \end{equation*}
Here, the search space is additionally restricted by the fixed size total causal effect $\mathcal{C}(i\rightarrow j)=\psi$ and given by $\R^{-G,\psi}:=\{T \in \R^{d \times d} : t_{k,l}=0 \text{ if } l \notin d(k), t_{i,j}=-\psi \}$. Solving this linear least squares problem with the constraint of one fixed parameter leads to the solution
    \begin{equation}\label{eq:constrdual}
       \sum_{k\neq i}^{d}(\widehat{\Sigma}^{-1})_{k,k|d(k)} + (\widehat{\Sigma}^{-1})_{i,i|d(i)\setminus \{j\}}+ \psi^2 (\widehat{\Sigma}^{-1})_{j,j|d(i)\setminus \{j\}} \\  + 2\psi (\widehat{\Sigma}^{-1})_{i,j|d(i)\setminus \{j\}}. 
    \end{equation}
Note that maximizing dual likelihood under constraints on total causal effects corresponds with the reciprocal property \eqref{eq:equal} to maximizing classical Gaussian likelihood under modified constraints on direct causal effects. Thus, while directly maximizing the classical Gaussian likelihood under constraints on the total causal effect is complicated due to the arising polynomial constraint on the parameters, we are able to obtain a closed-form solution for the maximum dual likelihood estimate since the constraint only pertains to one parameter.

\section{Computation of Confidence Regions for Total Causal Effects}\label{section:computation}

In the following section, we present important computational shortcuts as well as the results of a simulation study that analyzes the performance and computation times of our proposed confidence regions. 
     
\subsection{Computational Shortcuts}\label{section:shortcuts}

Our proposal of employing the dual likelihood to construct confidence regions for total causal effects resolves the need for numerical optimization routines of earlier methods and leads to a closed-form solution that significantly reduces computation times (see Figure \ref{fig:time2}). However, the bottleneck of the task itself is the superexponential growth of the number of possible causal structures with the number of nodes. Without any knowledge about the underlying causal relation, we have to consider all possible causal structures in order to make calibrated confidence statements that include the structure uncertainty. This, already at the scale of systems with $12$ involved variables, entails accounting for more than $10^{26}$ possible graphs. Using combinatorial shortcuts that quickly reduce the search space of plausible causal structures is thus extremely important to compute confidence regions that solve this task. 

First, we highlight that with our methodology, we only need to consider complete DAGs, thus reducing the search space to $d$ factorial structures for $d$ involved variables. Namely, we have to search through all permutations of $d$ nodes, which then correspond to all possible topological orderings of the involved variables. Furthermore, we employ a branch and bound type search algorithm to quickly reduce the space of plausible causal orderings before performing the testing procedure for all possible causal effects $\psi \in \R$. The main idea of this procedure in order to immediately reject implausible causal orderings is the following. We quickly approximate the maximum dual likelihood under the alternative $\Omega \in \mathcal{M}^{-1}$ with the dual likelihood estimate \eqref{eq:unconstrdual} under the causal ordering obtained via recursively minimizing $(\widehat{\Sigma}^{-1})_{k,k|d(k)}$ starting from the sink node. Then, recursively starting from the sink node, we search through the space of causal orderings in reverse and reject a partial ordering if the unrestricted partial dual likelihood estimate \eqref{eq:unconstrdual} already exceeds the threshold of the alternative. This is possible because the dual likelihood estimate \eqref{eq:unconstrdual} collapses into subproblems which only depend on descendants. Moreover, looking at \eqref{eq:constrdual}, the specific ordering among the descendants of $i$ is not relevant for the subsequent testing procedure for the precise total causal effects $\psi \in \R$. Therefore, in our branch and bound type search algorithm, it suffices to proceed with the partial ordering that achieves the highest dual likelihood among all plausible orderings with the same set of descendants $d(i)$. More specifically, at each step of our algorithm, we reject all partial orderings that do not achieve the highest dual likelihood among all plausible partial orderings of the same set of nodes and do not include $i$.

Combined, all these computational shortcuts vastly improve computation times and lead to a bottom-up procedure that allows us to compute the proposed confidence region in a reasonable time for moderate dimensions, see Section \ref{section:simulation}. We emphasize that considering the fast superexponential growth of the number of possible DAGs with the number of nodes, rigorously accounting for structure uncertainty over all DAGs is an intrinsically difficult task already in moderate dimensions. Further, we highlight the interesting distinction to the $\texttt{LRT}$-method of \citet{strieder:2023}, that by employing dual likelihood, our method leads to a bottom-up procedure, starting from the sink node, while their method is a top-down procedure, starting from the source node. 

\subsection{Simulation Study}\label{section:simulation}

In the following simulation study, we compare the width of the proposed confidence region as well as the computation times to the \texttt{LRT}-method by \citet{strieder:2023}. For more simulation results, we refer the reader to the Appendix \ref{appendix:simulation}. Our experiments are designed as follows. We generate $1000$ synthetic data sets based on linear SCMs corresponding to randomly selected DAGs. The edge weights are sampled from a normal distribution $N(\beta, 0.1)$, and the error terms are sampled from a standard normal distribution. Then we compute the confidence regions for the total causal effect $\mathcal{C}(1 \rightarrow 2)$ with confidence level $\alpha=0.05$. We repeat this procedure for a range of different average direct effect strengths $\beta$, sample sizes $n$, dimensions $d$, sparse or dense graphs, and for true non-zero total effects as well as for no total effect. 

First, we note that the theoretical coverage guarantees of our proposed confidence regions are based on (conservative) asymptotic critical values. Nevertheless, the confidence regions achieve the desired coverage in all our simulation settings, even in low sample sizes. We explore this in Table \ref{tab:cover} in the Appendix \ref{appendix:simulation}, which shows the empirical coverage frequencies against the sample size. Furthermore, while the employed dual likelihood estimators are asymptotically efficient in the sense of having the same asymptotic variance as ordinary maximum likelihood estimates, differences in finite samples are possible. Thus, we compare the performance of our confidence regions to the existing maximum likelihood approach. In Figure \ref{fig:width}, we compare the average width of the non-zero part of the confidence regions, reflecting the remaining uncertainty about the numerical size of the effect. We show the results for $d=10$ and $\beta=0.5 $ against the sample size. The main observation is that the difference in performance between both methods seems negligible. Similar observations can be made for other performance measures, which we defer to the Appendix \ref{appendix:simulation}. For example, we investigate how often the zero is included in the confidence regions when there is a true nonzero effect and observe no significant differences. Thus, both methods seem to be equally conclusive about the existence as well as the numerical size of the total causal effects. Further, the two methods also perform similarly under model misspecification of equal error variances and linearity.

    \begin{figure}[t]
    \centering
    \includegraphics[width=0.8\linewidth]{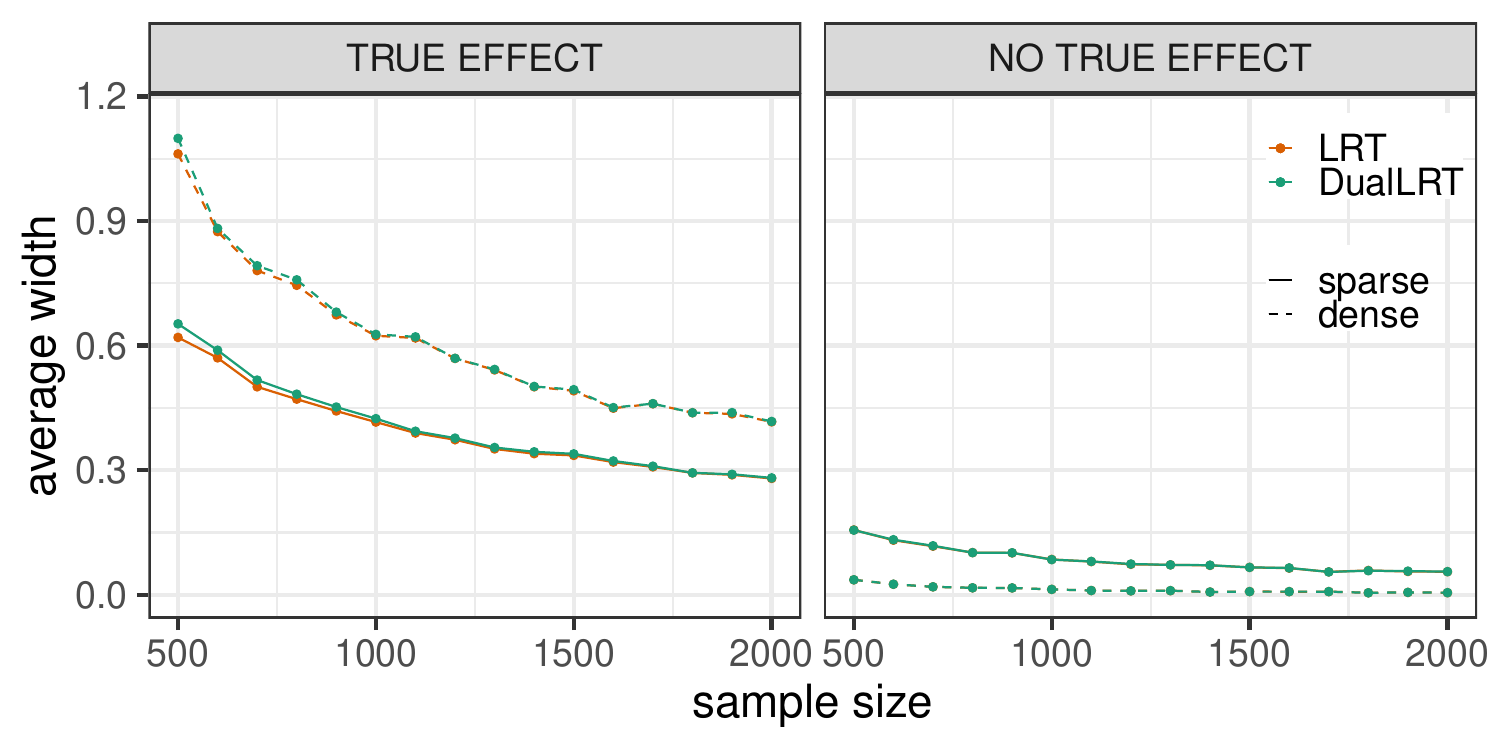}
    \caption{Mean width of $95\%$-confidence regions for the total causal effect in randomly generated $10$-dim. DAGs (1000 replications).}\label{fig:width}
    \end{figure}

However, the big advantage of the proposed dual likelihood method is the significantly reduced computation time. Employing the dual likelihood avoids time-consuming numerical optimization routines, and we immensely benefit from the closed-form solution calculated in Theorem \ref{theorem:duallrtinterval}. Figure \ref{fig:time1} shows the computation times for $\beta=0.5 $ and sample size $1000$ against the dimension. Our proposed confidence regions are faster to compute by a factor of up to $10^2$, without any compromises in terms of information value. This is even more apparent in Figure \ref{fig:time2}, which compares the computation times for $\beta=0.1 $ and sample size $1000$ against the dimension. In this setting, the generated data encompasses more structure uncertainty, given the lower average edge strength. Thus, the time difference between repeatedly applying numerical optimization algorithms versus directly computing the closed-form solution is more fatal, such that we are not able to compute the $\texttt{LRT}$-confidence regions in a reasonable time in dimensions beyond $d=8$. In contrast, the computation times of our proposed method seem to vary less with the amount of structure uncertainty inherent in the data set.

    \begin{figure}[t]
    \centering
    \includegraphics[width=0.82\linewidth]{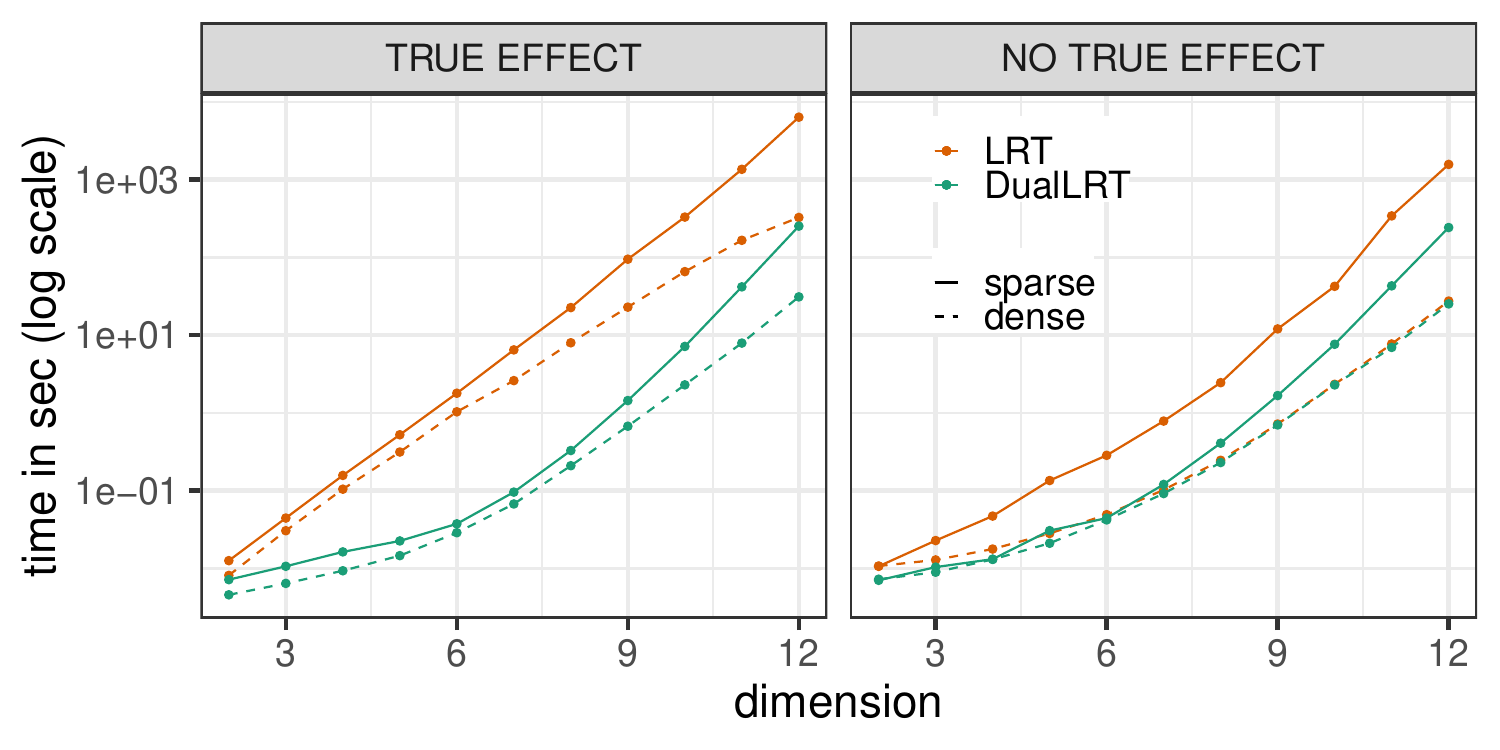}
    \caption{Mean computation times of $95\%$-confidence regions for the total causal effect in randomly generated DAGs with average edge strength $\beta = 0.5$ (1000 replications). }\label{fig:time1}
    \end{figure}
    
    \begin{figure}[t]
    \centering
    \includegraphics[width=0.82\linewidth]{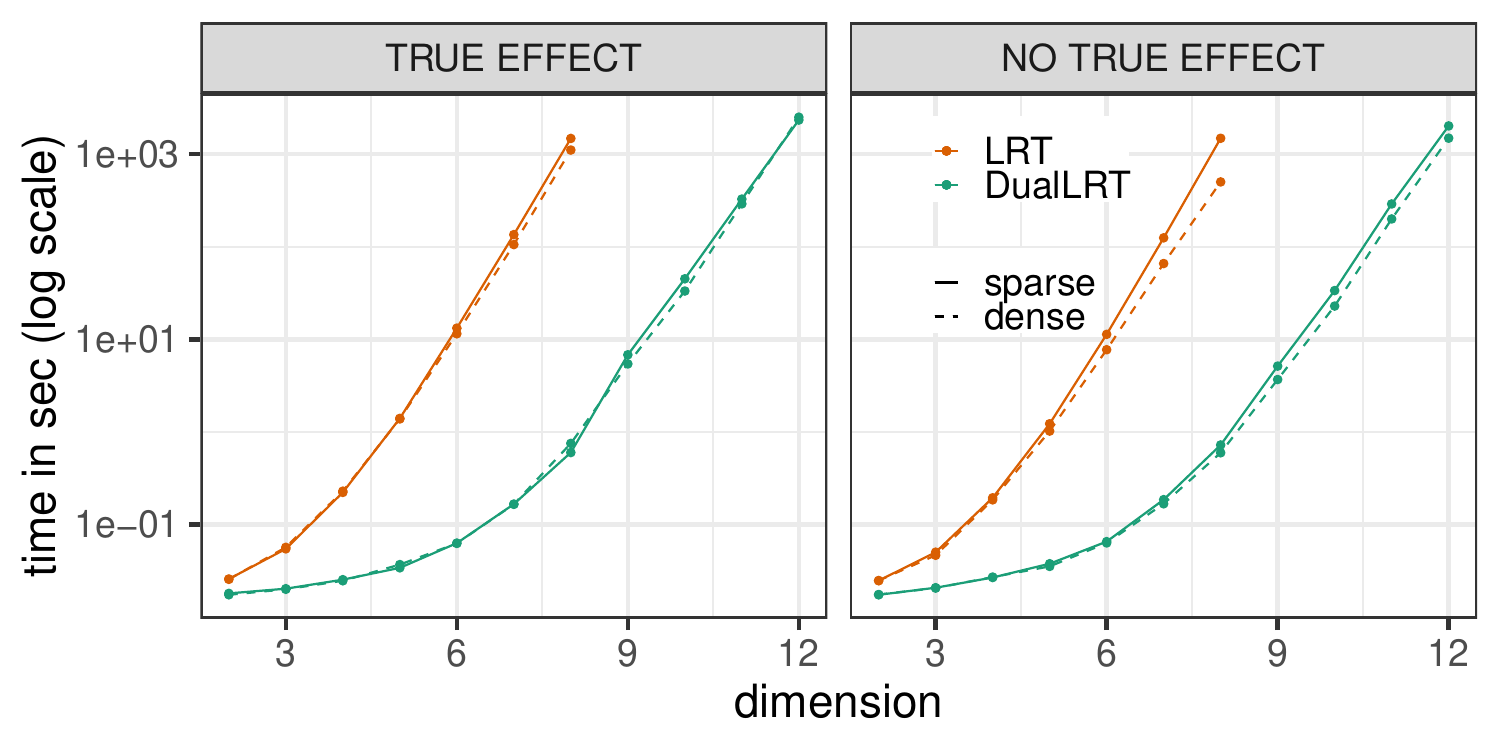}
    \caption{Mean computation times of $95\%$-confidence regions for the total causal effect in randomly generated DAGs with average edge strength $\beta = 0.1$ (1000 replications). }\label{fig:time2}
    \end{figure}

Furthermore, we highlight the observation that dense systems generally lead to wider confidence regions compared to sparse systems. Intuitively, for dense graphs more paths contribute to the total causal effects on average, and thus, more uncertainty about the numerical size of the effect remains in the data. In contrast, data generated by sparse DAGs generally exhibits more uncertainty about the underlying causal structure, which explains the observed increased computation times.

In summary, our simulation study validates that our proposed confidence regions based on the dual likelihood successfully pick up on the direction and the numerical size of the total causal effect while correctly quantifying the remaining uncertainty in structure as well as effect size. Furthermore, the computation times of our proposed confidence regions, computed with the proposed bottom-up procedure, are significantly lower than the competition. Here, we immensely benefit from the available closed-form solution due to employing dual likelihood. Especially under high structure uncertainty, avoiding repeated numerical optimization algorithms is crucial in order to compute the confidence regions in a reasonable time.

\section{Conclusion}\label{section:discussion}

In this article, we address the important problem of rigorously accounting for uncertainty in causal inference when the causal structure itself has to be learned from data. We provide a closed-form solution for constructing confidence regions for the total causal effect that captures structure uncertainty as well as uncertainty about the numerical size of the effect.  The confidence regions are developed for the concrete context of linear SCMs with equal error variances, where the underlying assumption of homoscedasticity across all interacting variables renders the causal structure identifiable. Thus, causal inference of the total causal effect without knowledge of the underlying structure is a well-defined task. An interesting conclusion of our work is that even in this specialized setting, we already observe high levels of uncertainty about causal directions and size of effects, which highlights the importance of methods that  properly account for structure uncertainty in causal inference.

One thing to note is that our general approach of leveraging test inversions of joint tests for causal structure and effect size is generalizable to other settings. In particular, if one deals with parametric models for which the causal DAG is identifiable, then likelihood ratio tests may be deployed with the same stochastic approximation strategies as in this paper, and we anticipate that the matrix inversion interplay between direct and total effects that is behind our use of dual likelihood can be similarly exploited.

While unique structure identifiability is required in order to consistently resolve uncertainty about nonzero causal effects (e.g., when facing full Markov equivalence in linear Gaussian models without variance assumption, the confidence region would always include zero), our framework could in principle be adjusted to target all effects that correspond to graphs within a Markov equivalence class. However, in cases where the causal structure is merely identifiable up to some faithfulness-type assumption, more research needs to be done to fully understand which causal effects are entailed by a given (possibly unfaithful) distribution.

The observed reciprocal correspondence between dual likelihood with total effects and classical likelihood with direct effects could be exploited for other causal inference tasks. The idea is to leverage tools developed for direct effects and classical likelihood to causal inference results for total effects via the dual likelihood. In our simulation study, the differences in performance and explanatory power that stem from employing dual likelihood ratio tests instead of classical likelihood ratio tests are negligible and heavily outweighed by the significantly reduced computation times due to the available closed-form solution.

Finally, given the fast superexponential growth of the number of possible causal structures with the dimension of the system, it is essential to quickly reject implausible orderings and reduce the search space in order to compute confidence regions that account for uncertainty over all structures. The introduced computational shortcuts lead to a bottom-up procedure that, starting from sink nodes, recursively reduces the search space by immediately rejecting implausible partial orderings. In contrast, using the classical likelihood leads to a top-down procedure that starts with source nodes. Cleverly alternating between both procedures in a first step to reduce the search space might even further decrease the computation times of our proposed confidence regions.

\begin{acks}[Acknowledgments]
This project has received funding from the European Research Council (ERC) under the European Union’s
Horizon 2020 research and innovation programme (grant agreement No. 83818). Further, this work has been
funded by the German Federal Ministry of Education and Research and the Bavarian State Ministry for
Science and the Arts. The authors of this work take full responsibility for its content.
\end{acks}

\bibliographystyle{imsart-nameyear}
\bibliography{bibliography.bib}

\begin{thebibliography}{13}

\bibitem[\protect\citeauthoryear{Brown}{1986}]{brown:1986}
\begin{bbook}[author]
\bauthor{\bsnm{Brown},~\bfnm{Lawrence~D.}\binits{L.~D.}}
(\byear{1986}).
\btitle{Fundamentals of statistical exponential families with applications in statistical decision theory}.
\bseries{Institute of Mathematical Statistics Lecture Notes---Monograph Series}
\bvolume{9}.
\bpublisher{Institute of Mathematical Statistics, Hayward, CA}.
\bmrnumber{882001}
\end{bbook}
\endbibitem

\bibitem[\protect\citeauthoryear{Casella and Berger}{1990}]{casella:book}
\begin{bbook}[author]
\bauthor{\bsnm{Casella},~\bfnm{George}\binits{G.}} \AND \bauthor{\bsnm{Berger},~\bfnm{Roger~L.}\binits{R.~L.}}
(\byear{1990}).
\btitle{Statistical inference}.
\bseries{The Wadsworth \& Brooks/Cole Statistics/Probability Series}.
\bpublisher{Wadsworth \& Brooks/Cole Advanced Books \& Software, Pacific Grove, CA}.
\bmrnumber{1051420}
\end{bbook}
\endbibitem

\bibitem[\protect\citeauthoryear{Chen, Drton and Wang}{2019}]{chen:2019}
\begin{barticle}[author]
\bauthor{\bsnm{Chen},~\bfnm{Wenyu}\binits{W.}}, \bauthor{\bsnm{Drton},~\bfnm{Mathias}\binits{M.}} \AND \bauthor{\bsnm{Wang},~\bfnm{Y.~Samuel}\binits{Y.~S.}}
(\byear{2019}).
\btitle{On causal discovery with an equal-variance assumption}.
\bjournal{Biometrika}
\bvolume{106}
\bpages{973--980}.
\bmrnumber{4031210}
\end{barticle}
\endbibitem

\bibitem[\protect\citeauthoryear{Drton}{2009}]{drton:2009}
\begin{barticle}[author]
\bauthor{\bsnm{Drton},~\bfnm{Mathias}\binits{M.}}
(\byear{2009}).
\btitle{Likelihood ratio tests and singularities}.
\bjournal{Ann. Statist.}
\bvolume{37}
\bpages{979--1012}.
\bmrnumber{2502658}
\end{barticle}
\endbibitem

\bibitem[\protect\citeauthoryear{Ghoshal and Honorio}{2018}]{ghoshal:2018}
\begin{binproceedings}[author]
\bauthor{\bsnm{Ghoshal},~\bfnm{Asish}\binits{A.}} \AND \bauthor{\bsnm{Honorio},~\bfnm{Jean}\binits{J.}}
(\byear{2018}).
\btitle{Learning linear structural equation models in polynomial time and sample complexity}.
In \bbooktitle{Proceedings of the Twenty-First International Conference on Artificial Intelligence and Statistics}
\bvolume{84}
\bpages{1466--1475}.
\bpublisher{PMLR}.
\end{binproceedings}
\endbibitem

\bibitem[\protect\citeauthoryear{Kauermann}{1996}]{kauermann:1996}
\begin{barticle}[author]
\bauthor{\bsnm{Kauermann},~\bfnm{G\"{o}ran}\binits{G.}}
(\byear{1996}).
\btitle{On a dualization of graphical {G}aussian models}.
\bjournal{Scand. J. Statist.}
\bvolume{23}
\bpages{105--116}.
\bmrnumber{1380485}
\end{barticle}
\endbibitem

\bibitem[\protect\citeauthoryear{Maathuis et~al.}{2019}]{handbook}
\begin{bbook}[author]
\beditor{\bsnm{Maathuis},~\bfnm{Marloes}\binits{M.}}, \beditor{\bsnm{Drton},~\bfnm{Mathias}\binits{M.}}, \beditor{\bsnm{Lauritzen},~\bfnm{Steffen}\binits{S.}} \AND \beditor{\bsnm{Wainwright},~\bfnm{Martin}\binits{M.}}, eds.
(\byear{2019}).
\btitle{Handbook of graphical models}.
\bseries{Chapman \& Hall/CRC Handbooks of Modern Statistical Methods}.
\bpublisher{CRC Press, Boca Raton, FL}.
\bmrnumber{3889064}
\end{bbook}
\endbibitem

\bibitem[\protect\citeauthoryear{Pearl}{2009}]{pearl:book}
\begin{bbook}[author]
\bauthor{\bsnm{Pearl},~\bfnm{Judea}\binits{J.}}
(\byear{2009}).
\btitle{Causality},
\bedition{Second} ed.
\bpublisher{Cambridge University Press, Cambridge}
\bnote{Models, reasoning, and inference}.
\bmrnumber{2548166}
\end{bbook}
\endbibitem

\bibitem[\protect\citeauthoryear{Peters and B\"{u}hlmann}{2014}]{peters:2014}
\begin{barticle}[author]
\bauthor{\bsnm{Peters},~\bfnm{Jonas}\binits{J.}} \AND \bauthor{\bsnm{B\"{u}hlmann},~\bfnm{Peter}\binits{P.}}
(\byear{2014}).
\btitle{Identifiability of {G}aussian structural equation models with equal error variances}.
\bjournal{Biometrika}
\bvolume{101}
\bpages{219--228}.
\bmrnumber{3180667}
\end{barticle}
\endbibitem

\bibitem[\protect\citeauthoryear{Peters, Janzing and Sch\"{o}lkopf}{2017}]{peters:book}
\begin{bbook}[author]
\bauthor{\bsnm{Peters},~\bfnm{Jonas}\binits{J.}}, \bauthor{\bsnm{Janzing},~\bfnm{Dominik}\binits{D.}} \AND \bauthor{\bsnm{Sch\"{o}lkopf},~\bfnm{Bernhard}\binits{B.}}
(\byear{2017}).
\btitle{Elements of causal inference}.
\bseries{Adaptive Computation and Machine Learning}.
\bpublisher{MIT Press, Cambridge, MA}.
\bmrnumber{3822088}
\end{bbook}
\endbibitem

\bibitem[\protect\citeauthoryear{Spirtes, Glymour and Scheines}{2000}]{spirtes:book}
\begin{bbook}[author]
\bauthor{\bsnm{Spirtes},~\bfnm{Peter}\binits{P.}}, \bauthor{\bsnm{Glymour},~\bfnm{Clark}\binits{C.}} \AND \bauthor{\bsnm{Scheines},~\bfnm{Richard}\binits{R.}}
(\byear{2000}).
\btitle{Causation, prediction, and search}.
\bseries{Adaptive Computation and Machine Learning}.
\bpublisher{MIT Press, Cambridge, MA}.
\bmrnumber{1815675}
\end{bbook}
\endbibitem

\bibitem[\protect\citeauthoryear{Strieder and Drton}{2023}]{strieder:2023}
\begin{barticle}[author]
\bauthor{\bsnm{Strieder},~\bfnm{David}\binits{D.}} \AND \bauthor{\bsnm{Drton},~\bfnm{Mathias}\binits{M.}}
(\byear{2023}).
\btitle{Confidence in causal inference under structure uncertainty in linear causal models with equal variances}.
\bjournal{J. Causal Inference}
\bvolume{11}.
\bmrnumber{4679813}
\end{barticle}
\endbibitem

\bibitem[\protect\citeauthoryear{Strieder et~al.}{2021}]{strieder:2021}
\begin{binproceedings}[author]
\bauthor{\bsnm{Strieder},~\bfnm{David}\binits{D.}}, \bauthor{\bsnm{Freidling},~\bfnm{Tobias}\binits{T.}}, \bauthor{\bsnm{Haffner},~\bfnm{Stefan}\binits{S.}} \AND \bauthor{\bsnm{Drton},~\bfnm{Mathias}\binits{M.}}
(\byear{2021}).
\btitle{Confidence in causal discovery with linear causal models}.
In \bbooktitle{Proceedings of the Thirty-Seventh Conference on Uncertainty in Artificial Intelligence. UAI'21}
\bvolume{161}
\bpages{1217--1226}.
\bpublisher{PMLR}.
\end{binproceedings}
\endbibitem

\end{thebibliography}

\clearpage
\appendix

\section{Additional Simulation Results}\label{appendix:simulation}

In this section, we present additional results of our simulation study. The data generation process and general simulation setup are outlined in Section \ref{section:simulation}. Table \ref{tab:cover} shows the observed coverage probabilities for $10$-dimensional graphs. Both methods seem to be conservative and achieve the desired coverage frequency. Considering that we employ critical values based on an upper bound as well as the theory of intersection union tests for the arising testing problem, it is not surprising that the resulting confidence regions are conservative. Further, note that under high structure uncertainty, it is not possible to compute the competitor \texttt{LRT} in a reasonable time, which we denoted with NA.

As an additional performance measure,  in Figure \ref{fig:zeros}, we compare the proportions of times zero is included in the confidence regions when there is a true nonzero effect. The inclusion of zero in the confidence regions reflects the remaining uncertainty about the existence of the effect. We show the results for $d=10$ and $\beta = 0.5$ against the sample size $n$ as well as for $d=10$ and sample size $n=1000$ against the average direct effect strength $\beta$. In terms of performance, we observe no significant differences. However, we note again that for high structure uncertainty, which is the case for low average direct effect strength $\beta$, it is not possible to compute the competitor \texttt{LRT} in a reasonable time.

Furthermore, with practical applications in mind, Figure \ref{fig:eqvar} investigates the robustness of the methods towards small deviations from equal error variances. Here, to generate data, we sample the error variances uniformly from $[1-0.5v, 1+0.5v]$, where $v$ indicates the degree of deviation from homoscedasticity across all variables. We show the empirical coverage probabilities for $d=10, \beta=0.5$, and sample size $n=1000$ against the degree of deviation $v$. Once again, we only observe negligible differences, and both methods indicate some robustness to small deviations from equal error variances.

Finally, to further explore aspects of model misspecification, we also investigate the robustness of the methods towards deviations from linear relations. We generate data as outlined before, however, with an increasing emphasis on an additional quadratic dependence structure via the relation $X_j= \sum_{i \in p(j)} \beta_i ((1+3v) X_i + 0.02 v X_i^2) + \varepsilon_j $, for all $j= 1 , \dots, d$. Thus, $v$ indicates the degree of deviation from linearity. The scaling ensures similar magnitudes of true causal effects across all investigated simulation settings. Note that for nonlinear relations, generally,  the true total causal effect is not given by a single parameter but rather by a functional relation. However, when applying linear models it is natural to consider as a parameter of interest the slope parameter determining the best linear approximation within the true causal structure.  
In Figure~\ref{fig:nonlinear}, we consider this `pseudo-linear' total effect and show the empirical coverage probabilities for $d=10, \beta=0.5$, and sample size $n=1000$ against the degree of deviation $v$. We observe that even under nonlinearity, the methods pick up on the direction of the causal relation, and thus, for no causal effect, zero is always covered by the confidence intervals. Note that in the considered nonlinear additive equal variance setting, the causal ordering is still implied by ordering conditional variances. Further, the performance differences between both methods are negligible, and for nonzero effects, both methods indicate some robustness to small deviations from linearity.

    \begin{table}[t]
    \centering
    \caption{Empirical coverage of $95\%$-confidence regions for the total causal effect in randomly generated $10$-dim. DAGs ($1000$ replications).}
    \resizebox{0.65\linewidth}{!}{
    \begin{tabular}{ rr|rr|rr|rr|rr }
    \toprule
    &  & \multicolumn{ 4 }{c|}{ TRUE EFFECT }  & \multicolumn{ 4 }{|c}{ NO TRUE EFFECT } \\
     &  & \multicolumn{ 2 }{c|}{ SPARSE }  & \multicolumn{ 2 }{|c|}{ DENSE } & \multicolumn{ 2 }{|c|}{ SPARSE } & \multicolumn{ 2 }{|c}{ DENSE } \\
    method & $n$\raisebox{0.1cm}{$\setminus\beta$}  & 0.1 & 0.5  & 0.1  & 0.5 & 0.1 & 0.5  & 0.1 & 0.5 \\ 
    \midrule
    \multirow{3}{*}{DualLRT} & 500 & 1.00 & 1.00 & 1.00 & 1.00 & 1.00 & 1.00 & 1.00 & 1.00  \\
    & 1000 & 0.99 & 1.00 & 1.00 & 1.00 & 1.00 & 1.00 & 1.00 & 1.00  \\
    & 2000 & 1.00 & 0.99 & 0.99 & 1.00 & 1.00 & 1.00 & 1.00 & 1.00  \\
    \midrule
    \multirow{3}{*}{LRT} & 500 &  NA & 0.99 & NA & 1.00 & NA & 1.00 & NA & 1.00  \\
    & 1000 & NA & 1.00 & NA & 1.00 & NA & 1.00 & NA & 1.00  \\
    & 2000 & NA & 1.00 & NA & 1.00 & NA & 1.00 & NA & 1.00   \\
    \bottomrule
    \end{tabular}%
    }
    \label{tab:cover}
    \end{table}

    \begin{figure}[t]
    \centering
    \includegraphics[width=0.82\linewidth]{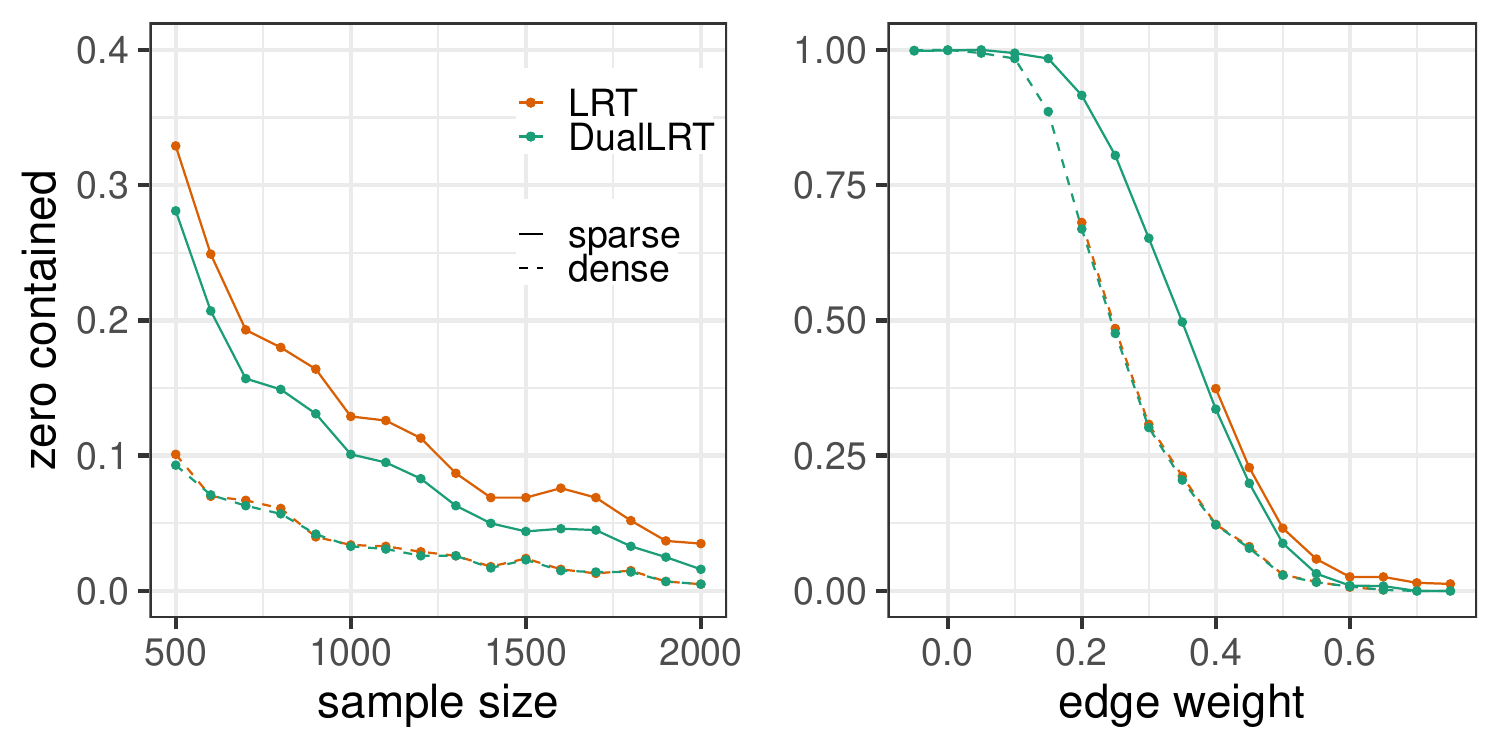}
    \caption{Proportion of times zero contained in $95\%$-confidence regions for the total causal effect in randomly generated $10$-dim. DAGs with true non-zero effect (1000 replications). }\label{fig:zeros}
    \end{figure}

    \begin{figure}[t]
    \centering
    \includegraphics[width=0.82\linewidth]{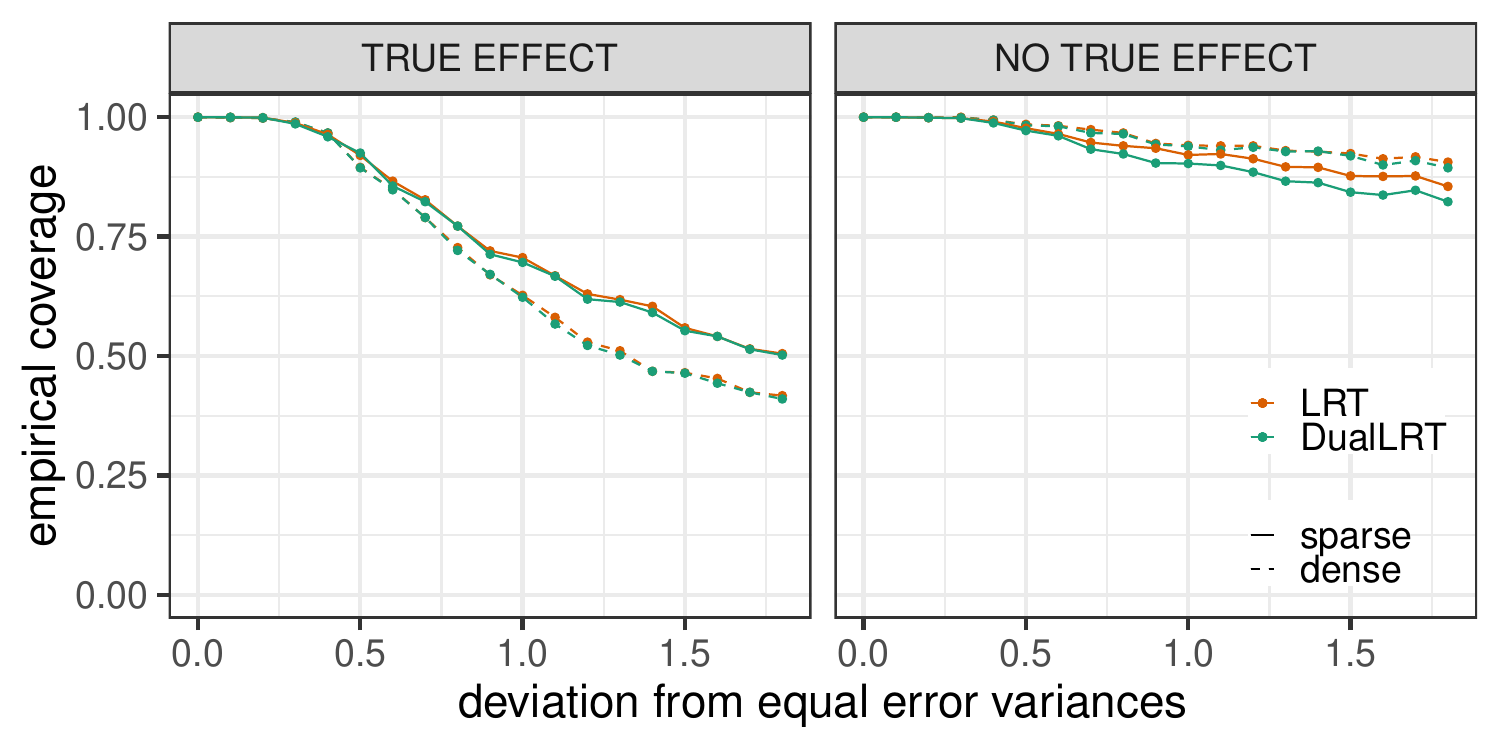}
    \caption{Empirical coverage of $95\%$-confidence regions for the total causal effect in randomly generated $10$-dim. DAGs under departure from equal error variances (1000 replications). }\label{fig:eqvar}
    \end{figure}

    \begin{figure}[t]
    \centering
    \includegraphics[width=0.82\linewidth]{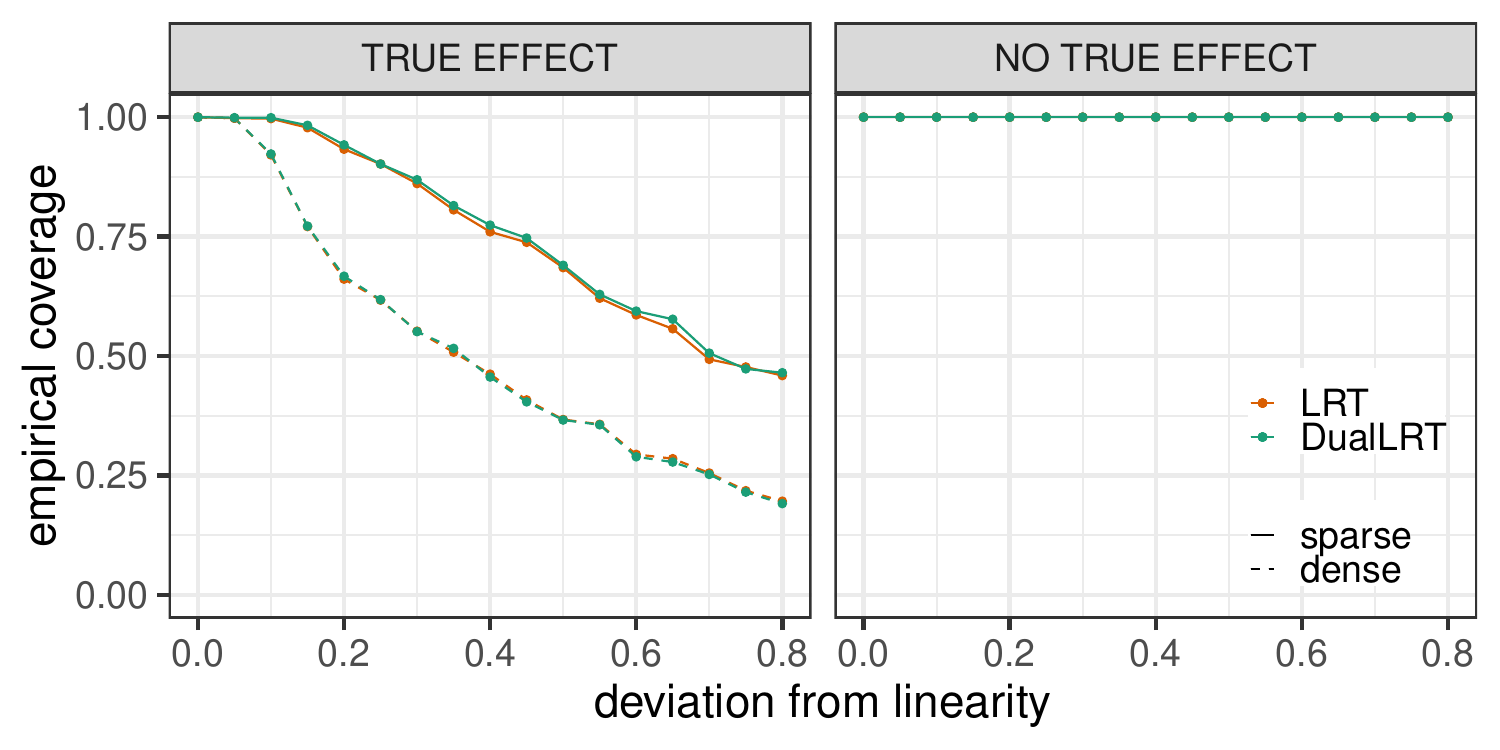}
    \caption{Empirical coverage of $95\%$-confidence regions for the total causal effect in randomly generated $10$-dim. DAGs under departure from linearity (1000 replications). }\label{fig:nonlinear}
    \end{figure}

\end{document}